%% file: distributed_state_estimator-report-arxiv.tex
\documentclass[a4paper]{article}
\usepackage{graphicx}
\usepackage{amsmath,amsthm}
\usepackage{amssymb, latexsym, mathrsfs}
\usepackage{enumerate}
\usepackage{algorithm}
\usepackage{subfigure}
\usepackage{color}
\usepackage{transparent}

\usepackage{url}
\usepackage[affil-it]{authblk}

\theoremstyle{plain}

\newtheorem{lem}{Lemma}
\newtheorem{prop}{Proposition}
\newtheorem{assum}{Assumption}

\theoremstyle{definition}
\newtheorem{defi}{Definition}

\newtheorem{prob}{Problem}

\theoremstyle{remark}
\newtheorem{rem}{Remark}

\input{macro}

\begin{document}

          \title{Distributed bounded-error state estimation for partitioned systems based on practical robust positive invariance\thanks{The research leading to these results has received funding from the European Union Seventh Framework Programme [FP7/2007-2013]  under grant agreement n$^\circ$ 257462 HYCON2 Network of excellence.}}

          \author{Stefano Riverso%
            \thanks{Electronic address: \texttt{stefano.riverso@unipv.it}; Corresponding author}} 

          \author{Daria Rubini%
            \thanks{Electronic address: \texttt{daria.rubini01@universitadipavia.it}}}

          \author{Giancarlo Ferrari-Trecate%
            \thanks{Electronic address: \texttt{giancarlo.ferrari@unipv.it}\\S. Riverso, D. Rubini and G. Ferrari-Trecate are with Dipartimento di Ingegneria Industriale e dell'Informazione, Universit\`a degli Studi di Pavia, via Ferrata 1, 27100 Pavia, Italy}}
          
          \affil{Dipartimento di Ingegneria Industriale e dell'Informazione\\Universit\`a degli Studi di Pavia}
          \date{\textbf{Technical Report}\\ November, 2013}

          \maketitle

          \begin{abstract}
            We propose a partition-based state estimator for linear discrete-time systems composed by coupled subsystems affected by bounded disturbances. The architecture is distributed in the sense that each subsystem is equipped with a local state estimator that exploits suitable pieces of information from parent subsystems. Moreover, differently from methods based on moving horizon estimation, our approach does not require the on-line solution to optimization problems. Our state-estimation scheme, that is based on the notion of practical robust positive invariance developed in \cite{Rakovic2011a}, also guarantees satisfaction of constraints on local estimation errors and it can be updated with a limited computational effort  when subsystems are added or removed.\\
            \emph{Keyword}: distributed state estimation, robust positive invariant sets, bounded error estimation.
          \end{abstract}

          \newpage

          \section{Introduction}
               In modern engineering there are several examples of applications composed by a large number of subsystems and for which centralized operations can be very expensive. For instance, the use of centralized controllers and state estimators can be hampered by the complexity of the design stage or by demanding computational and communication requirements for on-line operations. An alternative approach is to decompose the plant into physically coupled subsystems and design local controllers and state estimators associated to each subsystem. In these cases, local devices can operate in parallel using computational resources available at subsystem locations. Approaches with these features have been studied since the 1970's under the banner of decentralized and distributed control. 

               Available distributed state estimation schemes can be classified according to different criteria. First, the goal of a local state estimator can be either to reconstruct the state of the overall plant \cite{Alriksson2006,Farina2010a,Kamgarpour2008,Carli2008,Subbotin2009,Battistelli2011} or a subset of it \cite{Mutambara1998,Vadigepalli2003,Khan2008,Stankovic2009,Farina2010,Farina2011b,Haber2013}. In particular, estimators are termed partition-based if subsystems have non-overlapping states and a local estimator reconstructs the state of the corresponding subsystem only. Second, the topology of the communication network connecting local estimators can be different, ranging from all-to-all communication \cite{Vadigepalli2003} to transmission of information only from each subsystem to its children, i.e. subsystems influenced by it \cite{Khan2008,Stankovic2009,Farina2010,Farina2011b}. Third, local estimators can be based on unconstrained models \cite{Mutambara1998,Vadigepalli2003,Khan2008,Stankovic2009} or can cope with constraints on system variables such as disturbances, states \cite{Farina2010} and estimation errors \cite{Farina2011b}.

               In this paper we propose a novel partition-based state estimator for linear discrete-time subsystems affected by bounded disturbances. Similarly to the method proposed in \cite{Farina2010} and \cite{Farina2011b}, our scheme is distributed in the sense that computation of local state estimates can be performed in parallel but only after each estimator has received suitable pieces of information from parent subsystems. Moreover, as in \cite{Farina2011b}, state estimators account for constraints on subsystem disturbances and guarantee the fulfillment of \emph{a priori} specified constraints on local estimation errors. Differently from the scheme in \cite{Farina2010}, that is based on moving horizon estimation, and similarly to \cite{Farina2011b}, local estimators have a Luenberger structure and therefore do not require the on-line solution to optimization problems. Furthermore, most operations needed for the design of a local estimator can be performed using computational resources collocated with the corresponding subsystem and the only centralized step requires the analysis of a system whose order is equal to the number of subsystems. 

               In order to guarantee convergence of state estimates in absence of disturbances and fulfillment of prescribed constraints on the estimation error, we rely on the notion of practical robust positive invariance developed in \cite{Rakovic2011a} that is applied to the error dynamics. We also highlight that most of the appealing computational features of our method directly follow from results reported in \cite{Rakovic2011a} for the case of polytopic constraints. 
Since practical robust positive invariance implies worst-case robustness against the propagation of errors between subsystems, our design method involves some degree of conservatism. In the attempt of maximizing chances of successful design, we provide guidelines on the choice of local estimator parameters.  
We also show that when subsystems are added or removed, the state estimation scheme can be updated with limited efforts. 
More in detail, we prove that, in order to preserve convergence and fulfillment of constraints on estimation errors, (i) the plugging in of a subsystem requires the decentralized design of local estimators for the subsystem and its children only, besides the re-execution of the centralized step; (ii) the unplugging of a subsystem does not require any update.
               Compared to the distributed state estimator proposed in \cite{Farina2011b}, our scheme has several distinctive features. First, the use of the notion of practical robust positive invariance instead of the more standard concept of robust positive invariance, allows us to achieve a less conservative design procedure (see \cite{Rakovic2010a} for a discussion on the degree of conservativeness of various invariance concepts). Second, our local estimators can take advantage of the knowledge of parents' outputs and this can be fundamental for successful estimator design, and demonstrated in Section \ref{sec:example} through an example. Third, the method in \cite{Farina2011b} requires to analyze in a centralized fashion the stability of a system whose order is equal to the sum of the orders of all subsystems.

               The paper is structured as follows. Local state estimators are described in Section \ref{sec:distrstateesit}. In Section \ref{sec:pracRPI} we introduce practical robust decentralized invariance and show how it can be applied for guaranteeing convergence of estimators and constraint satisfaction. In Section \ref{sec:computational} we detail the design of local estimators. Section~\ref{sec:pnp} describes how to retune the estimator when subsystems are added or removed from the network.
In Section \ref{sec:example} we illustrate the use of the distributed state estimator for reconstructing the states of a power network system and compare our method with the state estimation scheme in \cite{Farina2011b}. Section \ref{sec:conclusion} is devoted to conclusions.

               \textbf{Notation.} We use $a:b$ for the set of integers $\{a,a+1,\ldots,b\}$. The symbol $\Rset_+^n$ stands for the vectors in $\Rset^n$ with nonnegative elements. The column vector with $s$ components $v_1,\dots,v_s$ is $\mbf v=(v_1,\dots,v_s)$. The symbol $\oplus$ denotes the Minkowski sum, i.e. $A=B\oplus C$ if and only if $A=\{a:a=b+c,~b\in B, ~c\in C\}$. Moreover, $\bigoplus_{i=1}^sG_i=G_1\oplus\ldots\oplus G_s$. The symbol $\One_\alpha$ (resp. $\Zero_\alpha$) denotes a matrix or a column vector with all $\alpha$ elements equal to $1$ (resp. $0$). Given a matrix $A\in\Rset^{n\times n}$, with entries $a_{ij}$ its entry-wise 1-norm is  $\norme{A}{1}=\sum_{i=1}^n\sum_{j=1}^n\abs{a_{ij}}$ and its Frobenious norm is $\norme{A}{F}^2=\sum_{i=1}^n\sum_{j=1}^na_{ij}^2$. Given a vector $x\in\Rset^n$ and a set $\Sset\subseteq\Rset^n$, $\dist{x}{\Sset}=\inf_{s\in\Sset}\norme{x-s}{}$. The pseudo-inverse of a matrix $A\in\Rset^{m\times n}$ is denoted with $A^\flat$.\\
               The set $\Xset\subseteq\Rset^n$ is Robust Positively Invariant (RPI) for $x(t+1)=f(x(t),w(t))$, $w(t)\in\Wset\subseteq\Rset^m$ if $x(t)\in\Xset\imply f(x(t),w(t))\in\Xset\mbox{, }\forall w(t)\in\Wset$. The RPI set $\bar\Xset$ is maximal if it includes every other RPI set. The set $\Xset\subseteq\Rset^n$ is positively invariant for $x(t+1)=f(x(t))$ if $x(t)\in\Xset\imply f(x(t))\in\Xset$.
               The set $\Xset\subseteq\Rset^n$ is a $\lambda$-contractive RPI set, with $\lambda\in[0,1)$ for $x(t+1)=f(x(t))$ if $x(t)\in\Xset\imply f(x(t))\in\lambda\Xset$.
A $\CC$-set is a set that is compact, convex and contains the origin.

          \section{Distributed State Estimator (DSE)}
              \label{sec:distrstateesit}
               We consider a discrete-time Linear Time Invariant (LTI) system
               \begin{equation}
                 \label{eq:model}
                 \begin{aligned}
                   \mbf\px &= \mbf{Ax+Bu+Dw}\\
                   \mbf y &= \mbf{Cx}
                 \end{aligned}
               \end{equation}
               where $\mbf x\in\Rset^n$, $\mbf u\in\Rset^m$, $\mbf y\in\Rset^p$ and $\mbf w\in\Rset^r$ are the state, the input, the output and the disturbance, respectively, at time $t$ and $\mbf\px$ stands for $\mbf x$ at time $t+1$. The state is partitioned into $M$ state vectors $\subss x i\in\Rset^{n_i}$, $i\in\MM=1:M$ such that $\mbf x=(\subss x 1,\ldots,\subss x M)$  and $n=\sum_{i\in\MM}n_i$. Similarly, the input, the output and the disturbance are partitioned into $M$ vectors $\subss u i\in\Rset^{m_i}$,  $\subss y i\in\Rset^{p_i}$,  $\subss w i\in\Rset^{r_i}$, $i\in\MM$ such that $\mbf u=(\subss u 1,\ldots,\subss u M)$, $m=\sum_{i\in\MM}m_i$, $\mbf y=(\subss y 1,\ldots,\subss y M)$, $p=\sum_{i\in\MM}p_i$, $\mbf w=(\subss w 1,\ldots,\subss w M)$ and $r=\sum_{i\in\MM}r_i$.

               We assume the dynamics of the $i$-th subsystem is given by
               \begin{equation}
                 \label{eq:subsystem}
                 \begin{aligned}
                   \subss\Sigma i:\quad\subss \px i&=A_{ii}\subss x i+B_i\subss u i+\sum_{j\in\NN_i}A_{ij}\subss x j+D_i\subss w i\\
                   \subss y i&=C_{i}\subss x i
                 \end{aligned}
               \end{equation}
               where $A_{ij}\in\Rset^{n_i\times n_j}$, $i,j\in\MM$, $B_i\in\Rset^{n_i\times m_i}$, $D_i\in\Rset^{n_i\times r_i}$, $C_i\in\Rset^{p_i\times n_i}$ and $\NN_i$ is the set of parents of subsystem $i$ defined as $\NN_i=\{j\in\MM:A_{ij}\neq 0, i\neq j\}$. Since $\subss y i$ depends on the local state $\subss x i$ only, subsystems $\subss\Sigma i$ are output-decoupled and then $\mbf{C}=\diag(C_1,\ldots,C_M)$. Similarly, subsystems $\subss\Sigma i$ are input- and disturbance-decoupled, i.e. $\mbf{B}=\diag(B_1,\ldots,B_M)$ and $\mbf{D}=\diag(D_1,\ldots,D_M)$. We also assume 
               \begin{equation}
                 \label{eq:Wbound}
                 \subss w i\in\Wset_i\subset\Rset^{r_i}
               \end{equation}

               In this section we propose a DSE for (\ref{eq:model}). We define for $i\in\MM$ the local state estimator
               \begin{equation}
                 \label{eq:subestimator}
                 \begin{aligned}
                   \subss{\tilde\Sigma}i:\quad\subss \tpx i =A_{ii}\subss \tx i+B_i\subss u i-L_{ii}(\subss y i-C_i\subss\tx i)+\\\sum_{j\in\NN_i}A_{ij}\subss \tx j-\sum_{j=1}^M\delta_{ij}L_{ij}(\subss y j-C_j\subss\tx j)\\
                 \end{aligned}
               \end{equation}
               where $\subss\tx i\in\Rset^{n_i}$ is the state estimate, $L_{ij}\in\Rset^{n_i\times p_j}$ are gain matrices and $\delta_{ij}\in\{0,1\}$. Hereafter we assume $\delta_{ij}=0$ and $L_{ij}=0$ if $j\not \in \NN_i$. This implies that $\subss{\tilde\Sigma} i$ depends only on local variables ($\subss\tx i$, $\subss u i$ and $\subss y i$) and parents' variables ($\subss\tx j$ and $\subss y j$, $j\in\NN_i$). Binary parameters $\delta_{ij}$, $j\in\NN_i$ can be chosen to take advantage of the knowledge of parents' outputs ($\delta_{ij}=1$) or to reduce the amount of information received form parents ($\delta_{ij}=0$).

               Defining the state estimation error as
               \begin{equation}
                 \label{eq:errori}
                 \subss e i=\subss x i-\subss\tx i,
               \end{equation}
               from (\ref{eq:subsystem}), \eqref{eq:subestimator} and (\ref{eq:errori}), we obtain the local error dynamics
               \begin{equation}
                 \label{eq:erroridyn}
                 \subss\pe i=\bar{A}_{ii}\subss e i+\sum_{j\in\NN_i}\bar{A}_{ij}\subss e j+D_i\subss w i \\
               \end{equation}
               where $\bar{A}_{ii}=A_{ii}+L_{ii}C_i$ and $\bar{A}_{ij}=A_{ij}+\delta_{ij}L_{ij}C_j$, $i\neq j$. Our main goal is to solve the following problem.
               \begin{prob}
                 \label{prob:estimator_properties}
                 Design local state estimators $\subss{\tilde\Sigma} i$, $i\in\MM$ that 
                 \begin{description}
                 \item[(a)] are nominally convergent, i.e. when $\Wset=\{0\}$ it holds
                   \begin{equation}
                     \label{eq:conv}
                     \norme{\subss e i(t)}{}\rightarrow 0 \mbox{ as  }t\rightarrow\infty 
                   \end{equation}
                 \item[(b)] guarantee 
                   \begin{equation}
                     \label{eq:bounderror}
                     \subss e i(t)\in\Eset_i,~\forall t\geq 0
                   \end{equation}
                   where $\Eset_i\subseteq\Rset^{n_i}$ are prescribed sets containing the origin in their interior. $\square$
                 \end{description}
               \end{prob}
               Defining the collective variable $\mbf e=(\subss e 1,\ldots,\subss e M)\in\Rset^n$, from (\ref{eq:erroridyn}) one obtains the collective dynamics of the estimation error 
               \begin{equation}
                 \label{eq:errordyn}
                 \begin{aligned}
                   \mbf\pe&=\mbf{\bar{A}}\mbf e+\mbf D\mbf w\\
                 \end{aligned}
               \end{equation}
               where the matrix $\mbf{\bar{A}}$ is composed by blocks $\bar{A}_{ij}$, $i,j\in\MM$.

               We equip system \eqref{eq:errordyn} with constraints $\mbf e\in\Eset=\prod_{i\in\MM}\Eset_i$ and $\mbf w\in\Wset=\prod_{i\in\MM}\Wset_i$. In Section~\ref{sec:pracRPI} we address Problem~\ref{prob:estimator_properties} under the following assumptions
               \begin{assum}
                 \label{ass:Foschur}
                 The matrices $\bar{A}_{ii}$, $i\in\MM$ are Schur.
               \end{assum}
               \begin{assum}
                 \label{ass:EWset}
                 The sets $\Eset_i$ and $\Wset_i$, $i\in\MM$ are $\CC$-sets.
               \end{assum}
               
               We highlight that if $\mbf L$ is such that $\mbf{\bar A}$ is Schur, then property \eqref{eq:conv} holds. If, in addition, Assumption~\ref{ass:EWset} holds, then there is an RPI set $\Omega\subset \Eset$ for the constrained system \eqref{eq:errordyn} (see \cite{Kolmanovsky1998}) and $\mbf{e}(0)\in\Omega$ guarantees property \eqref{eq:bounderror}. Remarkably, when sets $\Eset_i$ and $\Wset_i$ are polytopes, an RPI set $\Omega$ can be found solving a Linear Programming (LP) problem \cite{Rakovic2010}. However the LP problem includes the collective model \eqref{eq:model} in the constraints and computations become prohibitive for large $n$.
               
               In absence of coupling between subsystems (i.e. $A_{ij}=0$, $i\neq j$) the estimator dynamics \eqref{eq:subestimator} and error dynamics \eqref{eq:erroridyn} are decoupled as well. Therefore, under Assumptions~\ref{ass:Foschur} and \ref{ass:EWset}, properties \eqref{eq:conv} and \eqref{eq:bounderror} can be guaranteed computing RPI sets $\Omega_i\subseteq \Eset_i$ for each local error dynamics and requiring $\subss e i (0)\in\Omega_i$. Furthermore, if $\Eset_i$ and $\Wset_i$ are polytopes, the computation of sets $\Omega_i$, $i\in\MM$ amounts to the solution of $M$ LP problems that can be solved in parallel using computational resources collocated with subsystems. In order to propose a partially decentralized design procedure in presence of coupling between subsystems one has to take into account how coupling propagates errors between subsystems. As we will show in the next section, the notion of practical robust positive invariance, proposed in \cite{Rakovic2011a} allows one to study precisely this issue and offers a computationally feasible, yet conservative, procedure for solving Problem~\ref{prob:estimator_properties}.

          \section{Practical robust positive invariance for state estimation}
               \label{sec:pracRPI}

               In this section, we show how the main results of \cite{Rakovic2011a}, applied to the error dynamics \eqref{eq:erroridyn} equipped with constraints \eqref{eq:Wbound} and \eqref{eq:bounderror}, allow one to guarantee properties \eqref{eq:conv} and \eqref{eq:bounderror}.

               Given a collection of sets $\Sset=\{\Sset_i,~i\in\MM\}$, $\Sset_i\subset\Rset^{n_i}$ and a set $\Theta\subset\Rset_+^M$, we define a parameterized family of sets $\SSS(\Sset,\Theta)=\{(\theta_1\Sset_1,\ldots,\theta_M\Sset_M):\theta\in\Theta\}$, where $\theta=(\theta_1,\ldots,\theta_M)$. Intuitively, scalars $\theta_i$ can be interpreted as scaling factors. 
               \begin{defi}
                 The family of sets $\SSS(\Sset,\Theta)$ is practical Robust Positive Invariant (pRPI) for the constrained local error dynamics given by \eqref{eq:erroridyn}, \eqref{eq:Wbound} and \eqref{eq:bounderror}, if, for all $i\in\MM$ and all $(\theta_1\Sset_1,\ldots,\theta_M\Sset_M)\in\SSS(\Sset,\Theta)$, one has
                 \begin{subequations}
                   \label{eq:invariancefam}
                   \begin{align}
                     \theta_i\Sset_i&\subseteq\Eset_i  \label{eq:invariancefam-a}\\
                     \bar A_{ii}\theta_i\Sset_i\oplus\bigoplus_{j\in\NN_i}\bar A_{ij}\theta_j\Sset_j\oplus D_i\Wset_i&\subseteq\theta_i^+\Sset_i  \label{eq:invariancefam-b}\\
                     (\theta_1^+\Sset_1,\ldots,\theta_M^+\Sset_M)&\in\SSS(\Sset,\Theta)  \label{eq:invariancefam-c}
                   \end{align}
                 \end{subequations}
               \end{defi}

               \begin{assum}
                 \label{ass:Sset}
                 The sets $\Sset_i$, $i\in\MM$ are $\CC$-sets containing the origin in their interior.
               \end{assum}
               The main issue we will address in the sequel is the following: given $\Sset$, is there a nonempty set $\Theta\subset \Rset^M_+$ such that the family $\SSS(\Sset,\Theta)$ is pRPI ? In order to provide an answer, in \cite{Rakovic2011a} it is proposed to first derive the dynamics of the scaling factors $\theta_i$. More precisely, for  all $i,j\in\MM$ we set
               \begin{equation}
                 \label{eq:muij}
                 \mu_{ij} =\begin{cases}     
                   \min_{\substack{\mu\geq 0}}\{\mu:\bar A_{ij}\Sset_j\subseteq\mu\Sset_i\} & \mbox{ if }i=j\mbox{ or }j\in\NN_i\\
                   0 & \mbox{ otherwise}
                 \end{cases}
               \end{equation}
               \begin{equation}
                 \label{eq:alphai}
                 \alpha_i = \min_{\substack{\beta\geq 0}}\{\beta:D_i\Wset_i\subseteq\beta\Sset_i\}.
               \end{equation}
               and define the collective dynamics of the scaling factors
               \begin{equation}
                 \label{eq:thetadyn}
                 \theta^+=T\theta+\alpha
               \end{equation}
               where the entries of $T\in\Rset^{M\times M}$ are $T_{ij}=\mu_{ij}$ and $\alpha=(\alpha_1,\ldots,\alpha_M)$. 
               It is easy to show that \eqref{eq:thetadyn} guarantees
               \begin{equation}
                 \label{eq:error_property}
                 \subss e i\in\theta_i S_i \Rightarrow \subss e i^+\in\theta_i^+ S_i.
               \end{equation}
               For fulfilling \eqref{eq:invariancefam-a}, let us define
               \begin{equation}
                 \label{eq:theta0}
                 \Theta_0=\{\theta\in\Rset^M:\forall i\in\MM,~\theta_i\Sset_i\subseteq\Eset_i\}
               \end{equation}

               The key assumption used in \cite{Rakovic2011a} for providing a set $\Theta$ that makes $\SSS(\Sset,\Theta)$ a pRPI family is the following one.
               \begin{assum}
                 \label{ass:thetaass}
                 \begin{enumerate}[(i)]
                 \item\label{enu:Tschur} $T$ is Schur.
                 \item\label{enu:bartheta} The unique equilibrium point $\bar\theta$ of system (\ref{eq:thetadyn}) is such that $\bar\theta\in\Theta_0$.
                 \item\label{enu:Theta0} The set $\Theta$ is an invariant set for system (\ref{eq:thetadyn}) and constraint set $\Theta_0$, i.e. $\forall \theta\in\Theta\subseteq\Theta_0$, $\theta^+\in\Theta$.
                 \end{enumerate}
               \end{assum}
               
               \begin{lem}[\cite{Rakovic2011a}]
                 \label{lem:theoRak}
                 Let Assumptions \ref{ass:Foschur}-\ref{ass:thetaass} hold. Then,
                 \begin{description}
                 \item[(i)] there is a non-trivial convex and compact positively invariant set $\Theta$ for system \eqref{eq:thetadyn} equipped with constraints $\theta\in\Theta_0$; 
                 \item[(ii)] $\SSS(\Sset,\Theta)$ is pRPI for  \eqref{eq:erroridyn} with constraints \eqref{eq:Wbound} and \eqref{eq:bounderror}. $\square$
                 \end{description}
               \end{lem}

               Lemma~\ref{lem:theoRak} guarantees that
               \begin{equation}  
                 \label{eq:init}
                 \begin{aligned}
                   \theta(0)\in&\Theta \mbox{ and } \subss e i (0)\in\theta_i(0)\Sset_i,~ \forall i\in\MM \Rightarrow \\
                   & \subss e i (t)\in \theta_i(t)\Sset_i,~\forall i\in \MM, ~\forall t\geq 0 
                 \end{aligned}
               \end{equation}
               Furthermore, as shown in \cite{Rakovic2011a}, $\dist{\subss e i(t)}{\bar\theta_i\Sset_i}\rightarrow 0$ as $t\rightarrow\infty$. In the nominal case, i.e. $\Wset=\{0\}$, one has $\alpha=0$ in \eqref{eq:thetadyn}.  Then $\bar\theta=0$ and property \eqref{eq:conv} is guaranteed. Also \eqref{eq:bounderror} holds since, from \eqref{eq:init} and \eqref{eq:invariancefam-a} one has $\subss e i (t) \in\theta_i(t)\Sset_i\subseteq \Eset_i$. Therefore, Problem~\ref{prob:estimator_properties} is solved if we can design local state estimators fulfilling the assumptions of Lemma~\ref{lem:theoRak}. A design procedure to achieve this goal is proposed in Section~\ref{sec:computational}.
               \begin{rem}
                 \label{rmk:init}
                 Note that, according to \eqref{eq:init}, the initialization of the local estimators requires to find a suitable initial state $\theta (0)\in \Theta$ for system~\eqref{eq:thetadyn} and this is a  centralized operation. In order to allow each estimator to locally compute its initial state, one can build offline an inner box approximation $\bar \Theta=\prod_{i=1}^M [ \underline \theta _i,\bar \theta_i]$ contained in $\Theta$ and choose $ \subss{\tilde x} i (0) $  such that $\subss x i (0)-\subss{\tilde x} i (0) \in [ \underline \theta _i,\bar \theta_i]$.
               \end{rem}
               
          \section{Design of local estimators}
               \label{sec:computational}
               In this section, we propose a method to design the distributed state estimator presented in Sections \ref{sec:distrstateesit} and \ref{sec:pracRPI}. The key issue is how to compute suitable gains $L_{ij}$ and binary variables $\delta_{ij}$ such that Assumption~\ref{ass:thetaass} holds. From now on we consider polytopic sets $\Eset_i$, $\Wset_i$ and $\Sset_i$, $i\in\MM$ verifying Assumptions ~\ref{ass:EWset} and \ref{ass:Sset}. Without loss of generality we can write
               \begin{subequations}
                 \begin{align}
                   \label{eq:setspolyE}\Eset_i &= \{h_{i,\tau}^T\subss{e}{i}\leq 1,\forall \tau\in 1:\bar{\tau}_i\} = \{ \HH_i\subss{e}{i}\leq \One_{\bar{\tau}_i} \} \\
                   \label{eq:setspolyW}\Wset_i &= \{f_{i,\upsilon}^T\subss{w}{i}\leq 1,\forall \upsilon\in 1:\bar{\upsilon}_i\} = \{ \FF_i\subss{w}{i}\leq \One_{\bar{\upsilon}_i} \} \\
                   \label{eq:setspolyS}\Sset_i &= \{g_{i,\psi}^T\subss{s}{i}\leq 1,\forall \psi\in 1:\bar{\psi}_i\} = \{ \GG_i\subss{s}{i}\leq \One_{\bar{\psi}_i} \}
                 \end{align}
               \end{subequations}
               where $\HH_i=(h_{i,1}^T,\ldots,h_{i,\bar \tau_i}^T)\in\Rset^{\bar{\tau}_i\times n_i}$, $\FF_i=(f_{i,1}^T,\ldots,f_{i,\bar \upsilon_i}^T)\in\Rset^{\bar{\upsilon}_i\times r_i}$ and $\GG_i=(g_{i,1}^T,\ldots,g_{i,\bar \psi_i}^T)\in\Rset^{\bar{\psi}_i\times n_i}$. The design procedure is summarized in Algorithm 1 that is composed by three parts.               
               \begin{algorithm}[!htb]
                 \caption{}
                 \label{alg:practdec}
                 \textbf{Input}: polytopic sets $\Eset_i$, $\Wset_i$, $i\in\MM$ verifying Assumption~\ref{ass:EWset}.\\
                 \textbf{Output}: A pRPI family of sets $\SSS(\Sset,\Theta)$.\\
                 
                 \begin{enumerate}[(A)]
                 \item\label{alg:dec}\emph{Decentralized steps}. For all $i\in\MM$,
                   \begin{enumerate}[(I)]
                   \item\label{alg:step1dec} compute the matrix $L_{ii}$ such that $\bar A_{ii}$ is Schur and has as many zero eigenvalues as possible;
                   \item\label{alg:step2dec} compute a $\lambda_i$-contractive set $\Sset_i$ for 
                     \begin{equation}
                       \label{eq:errorilocal}
                       \subss \pe i = \bar A_{ii}\subss e i
                     \end{equation}
                     verifying $\Sset_i\subseteq\Eset_i$ and set $\mu_{ii}=\lambda_i$;                
                   \item\label{alg:step3dec} compute $\alpha_i$ as in \eqref{eq:alphai}.
                   \end{enumerate}
                   
                 \item\label{alg:dist}\emph{Distributed steps}. For all $i\in\MM$,
                   \begin{enumerate}[(I)]
                   \item\label{alg:step1dist} if $\delta_{ij}=1$, compute the matrix $L_{ij}$, $\forall j\in\NN_i$ solving
                     \begin{equation}
                       \label{eq:Lijopt}
                       \min_{\substack{L_{ij}}} \norme{\GG_i\bar A_{ij}\GG_j^\flat}{p}
                     \end{equation}
                     where either $p=1$ or $p=F$.
                   \item\label{alg:step2dist} compute $\mu_{ij}$ as in \eqref{eq:muij}.
                   \end{enumerate}
                   
                 \item\label{alg:cen}\emph{Centralized steps}
                   \begin{enumerate}[(I)]
                   \item\label{alg:step1cen} if matrix $T$ is not Schur \textbf{stop};
                   \item\label{alg:step2cen} compute set $\Theta_0$ as in \eqref{eq:theta0} and the equilibrium point $\bar\theta$ of system \eqref{eq:thetadyn}. If $\bar\theta\notin\Theta_0$ \textbf{stop};
                   \item\label{alg:step3cen} compute the maximal invariant set $\Theta_\infty$ of system \eqref{eq:thetadyn} equipped with constraint $\Theta_0$;
                   \item\label{alg:step4cen} compute an inner box approximation $\bar\Theta$ of $\Theta_\infty$.
                   \end{enumerate}
                 \end{enumerate}
               \end{algorithm}
               Operations in part (A) can be executed in parallel using computational resources associated with subsystems, i.e. in a decentralized fashion. Steps in part (B) have a distributed nature, meaning that computations are decentralized but they can be performed only after each system has received suitable pieces of information from its parents. Finally, design steps in part (C) require centralized computations involving only  the $M$-th order system \eqref{eq:thetadyn}. Next, we comment each step of Algorithm~\ref{alg:practdec} in details.
               
               \subsection{Part (A)}
                    \label{subs:partA}
                    Step (\ref{alg:step1dec}) is the easiest one and it can be performed only if pairs $(A_{ii},C_i)$, $i\in\MM$ are detectable. The requirement of placing eigenvalues of $\bar A_{ii}$ in zero is motivated by step (\ref{alg:step2dec}). 
               
                    The computation of sets $\Sset_i$ as in step  (\ref{alg:step2dec}) has been suggested in \cite{Rakovic2011a} and it is based on the argument that sets $(1-\lambda_i)$ can be used for compensating coupling terms in the error dynamics. Remarkably, using the efficient procedures proposed in \cite{Rakovic2010}, the computation of a set $\Sset_i$ amounts to solving the optimization problem
                    \begin{subequations}
                      \label{eq:lplambda}
                      \begin{align} 
                        \label{eq:lplambda1}\PP_i(\Sset_i^0,k_i):&~\min_{\gamma_i,\beta_i,\{\Sset_i^s\}_{s=1}^{k_i}}\gamma_i\\
                        &\label{eq:lplambda2}\gamma_i\in[0,1),\quad\Sset_i^{k_i}\subseteq\gamma_i\Sset_i^0\\
                        &\label{eq:lplambda3}\beta_i\in\Rset_+,\quad\bigoplus_{s=0}^{k_i-1}\Sset_i^s\subseteq\beta_i\Eset_i\\
                        &\label{eq:lplambda4}\Sset_i^{s}=\bar A_{ii}^{s}\Sset_i^{0},\mbox{ }\forall s=1,\ldots,k_i
                      \end{align}
                    \end{subequations} 
                    where $k_i\in \Nset$ and the set $\Sset_i^0\subset \Rset^{n_i}$ are provided as inputs. In particular, \eqref{eq:lplambda} is an LP problem and the set $\Sset_i$ can be obtained as $\Sset_i=\beta_i^{-1}\bigoplus_{s=0}^{k-1}\Sset_i^s$. Furthermore, the contractivity parameter is $\lambda_i=\frac{\delta_i+\gamma_i^*-1}{\delta_i}$, where $\gamma_i^*$ is a solution to \eqref{eq:lplambda} and $\delta_i=\min_{\tilde\delta}\{\tilde\delta:\bigoplus_{s=0}^{k_i-1}\Sset_i^s\subseteq\tilde\delta\Sset_i^0,\tilde\delta\geq 1\}$. Note that also $\delta_i$ can be computed solving an LP problem. As shown in \cite{Rakovic2010}, since the matrix $\bar A_{ii}$ is Schur, then, given a $\CC$-set $\Sset_i^0$, there exists a sufficiently large $k_i$ such that problem \eqref{eq:lplambda} is feasible. Moreover, if all eigenvalues  of $\bar A_{ii}$ are zero, feasibility of \eqref{eq:lplambda} can be guaranteed setting $k_i=n_i$. Indeed since $\bar A_{ii}^{n_i}=\Zero_{n_i\times n_i}$ we have $\Sset_i^{n_i}=\{0\}$ and hence, irrespectively of $\Sset_i^0$, constraints \eqref{eq:lplambda2} hold with $\alpha_i=0$. Moreover, since from~\eqref{eq:lplambda4} sets $\{\Sset_i^s\}_{s=1}^{k_i-1}$ are polytopes containing the origin, then there exists $\beta_i$ such that constraints \eqref{eq:lplambda3} hold. We highlight that the scalar $\mu_{ii}$ computed as in \eqref{eq:muij} is equal to the contractivity parameter $\lambda_i$. 

                    Step (\ref{alg:step3dec}) focuses on the computation of scalars $\alpha_i$. From \eqref{eq:alphai} and \eqref{eq:setspolyW}, using procedures proposed in \cite{Kolmanovsky1998}, we have $\alpha_{i} = \max_{\substack{\psi\in 1:\bar{\psi}_i}}\{z_i\}$ where
                    \begin{equation}
                      \label{eq:alphaicomp}
                      \begin{aligned}
                        z_{i} = &\max_{\substack{\subss w i}}~g_{i,\psi}D_i\subss w i\\
                        &\FF_i\subss w i\leq\One_{\bar{\upsilon}_i}
                      \end{aligned}
                    \end{equation}
                    Therefore, step  (\ref{alg:step3dec}) requires the solution to the $\psi_i$ LP problems \eqref{eq:alphaicomp}.

               \subsection{Part (B)}
                    For the computation of matrices $L_{ij}$ and parameters $\mu_{ij}$, each system $\subss\Sigma i$ needs to receive the matrix $C_j$ and the set $\Sset_j$ from parents $j\in\NN_i$ such that $\delta_{ij}=1$. 

                    In step (\ref{alg:step1dist}), if $\delta_{ij}=1$, the computation of matrices $L_{ij}$, $j\in\NN_i$ is required. Since the choice of $L_{ij}$ affects the coupling term $\bar A_{ij}$ and hence the Schurness of matrix $T$, we propose to reduce the magnitude of coupling by minimizing the magnitude of $\bar A_{ij}$ in \eqref{eq:Lijopt}, where $\GG_i$ and $\GG_j^\flat$ allow us to take into account the size of sets $\Sset_i$ and $\Sset_j$, respectively. More precisely, it can be shown that the term $\norme{\GG_i\bar A_{ij}\GG_j^\flat}{p}$ is a  measure of how much the coupling term  $\bar A_{ij}\subss s j$, $j\in\NN_i$  affects the fulfillment of the constraint $\subss s i\in\Sset_i$. We highlight that the minimization of  $\| \GG_i\bar A_{ij}\GG_j^\flat\|_1$ in \eqref{eq:Lijopt} amounts to an LP problem and the minimization of $\| \GG_i\bar A_{ij}\GG_j^\flat\|_F$ can be recast into a Quadratic Programming (QP) problem. So far the parameters $\delta_{ij}$ have been considered fixed. However, if in step (\ref{alg:step1dist}) one obtains $L_{ij}=0$ for some $j\in\NN_i$, it is impossible to reduce the magnitude of the coupling term $\bar A_{ij}$ and, from \eqref{eq:subestimator}, the knowledge of $\subss y j$ is useless. This suggests to revise the choice of $\delta_{ij}$  and set $\delta_{ij}=0$. In step (\ref{alg:step2dist}),  since $\Sset_i$ are polytopes, using procedures proposed in \cite{Kolmanovsky1998} we  can compute scalars $\mu_{ij}$ as
                    \begin{equation}
                      \label{eq:muijcomp}
                      \mu_{ij} = \max_{\substack{\psi\in 1:\bar{\psi}_i}}\{\max_{\substack{\subss s j}} g_{i,\psi}\bar A_{ij}\subss s j:\GG_j\subss s j\leq\One_{\bar{\psi}_j}\}.
                    \end{equation}
                    that requires the solution of $\bar{\psi}_i$ LP problems.

               \subsection{Part (C)}
                    In step (\ref{alg:step1cen}) we check the Schurness of matrix $T$. If the test fails, Assumption~\ref{ass:thetaass}-(i) cannot be fulfilled and the only possibility is to restart the algorithm after increasing the number of variables $\delta_{ij}$ that are equal to one.

                    In step (\ref{alg:step2cen}), since the sets $\Sset_i$ and $\Eset_i$ are polytopes, using results from \cite{Kolmanovsky1998} the computation of the set $\Theta_0$ can be done as follows 
                    \begin{equation}
                      \label{eq:Theta0comp}
                      \begin{aligned}
                        \Theta_0 &= \prod_{i=1}^M[0,\tilde \theta_i]\\
                        \tilde\theta_i &= (\max_{\substack{\tau\in 1:\bar{\tau}_i}}\{\sup_{\substack{\subss s i}} h_{i,\tau}\subss s i:\GG_i\subss s i\leq\One_{\bar{\psi}_i}\})^{-1}.
                      \end{aligned}
                    \end{equation}
                    Moreover, in step (\ref{alg:step2cen})  we compute the equilibrium point $\bar\theta$ of system \eqref{eq:thetadyn}. If $\bar\theta\notin\Theta_0$ we can not guarantee property~\eqref{eq:bounderror} and therefore the algorithm stops. Note that if $\Wset_i=\{0\}$, $\forall i\in\MM$, the equilibrium point $\bar\theta$ is the origin and hence $\bar\theta\in\Theta_0$ by construction. 

                    According to Assumption~\ref{ass:thetaass}-\ref{enu:Theta0} , the set $\Theta$ of all feasible contractions $\theta$ is computed as an RPI set for system \eqref{eq:thetadyn} and constraints $\theta\in\Theta_0$. In particular, since $T$ is Schur and $\Theta_0$ is a polytope, using results from \cite{Gilbert1991} we can compute the maximal RPI set $\Theta_{\infty}$ by solving a suitable LP problem.

                    As discussed in Remark~\ref{rmk:init}, a decentralized initialization of state estimators is possible computing  an hyperrectangle $\bar \Theta$ contained in $\Theta_\infty$. This is done in step (\ref{alg:step4cen}). More precisely, using results from \cite{Bemporad2004}, we can set $\bar\Theta=\prod_{i=1}^M[0,\bar \theta_i]$ where 
                    \begin{eqnarray}
                      \label{eq:Thetabox}
                      \bar \theta_i&=&\max_{\substack{\tilde\theta\in\Theta_{\infty}}}~\gamma^T\tilde\theta, \\
                      \gamma&=&(\gamma_1,\ldots,\gamma_M)\nonumber \\
                      \gamma_i &=& (\max_{\substack{\theta}}~\theta_i:\theta\in\Theta_\infty)^{-1}. \label{eq:gammaiLP}
                    \end{eqnarray}
                    
                    As described in  \cite{Bemporad2004}, the vector $\gamma$ is used for maximizing the volume of $\bar\Theta$. From \eqref{eq:Thetabox} and \eqref{eq:gammaiLP} the computation of the hyper-rectangle $\bar\Theta$ requires the solution of $M+1$ LP optimization problems.

          \section{Large-scale systems with variable number of subsystems}
\label{sec:pnp}
               In this section, we discuss the retuning of the DSE when a subsystem is added or removed. We highlight that plugging-in and unplugging of subsystems are here considered as offline operations. In particular, we will show how to preserve properties \eqref{eq:conv} and \eqref{eq:bounderror} without performing all computations required by Algorithm~\ref{alg:practdec}. As a starting point, we consider system \eqref{eq:model} equipped with a DSE designed using Algorithm~\ref{alg:practdec}.

               \subsection{Plug-in operation}
                    Assume the new subsystem $\subss\Sigma{M+1}$ is plugged in and set $\bar\MM=\MM\cup\{M+1\}$. Since the overall system has changed, in principle one has to design the DSE from scratch running Algorithm \ref{alg:practdec}. Note however that Part (A) of Algorithm~\ref{alg:practdec} is decentralized and therefore it has to be executed for the new subsystem only. Part (B) of Algorithm~\ref{alg:practdec} involves only the new subsystem, its parents and its children $\CC_{M+1}=\{j\in\MM: A_{M+1,j}\neq 0, j\neq M+1\}$. In fact, subsystem $\subss \Sigma {M+1}$ needs sets $\Sset_j$ from its parents for computing parameters $\mu_{M+1,j}$, $j\in \NN_{M+1}$. Moreover since children of $\subss \Sigma {M+1}$ have a new parent, they need to know $\Sset_{M+1}$ in order to update parameters $\mu_{k,M+1}$, $k\in\CC_{M+1}$.

                    If Step (\ref{alg:step1cen}) or Step (\ref{alg:step2cen}) fail, we declare that system $\subss\Sigma{M+1}$ can not be added, because the family of sets $\SSS(\Sset,\Theta)$ is not a pRPI. In Algorithm \ref{alg:plugin}  we summarize the computations for updating the DSE that are triggered by the addition of $\subss\Sigma{M+1}$.
                    \begin{algorithm}[!htb]
                      \caption{}
                      \label{alg:plugin}
                      \textbf{Input}: new subsystem $\subss\Sigma{M+1}$ with sets $\Eset_{M+1}$ and $\Wset_{M+1}$.\\
                      \textbf{Output}: an updated pRPI family of sets $\SSS(\Sset,\Theta)$.\\
                      
                      \begin{enumerate}[(A)]
                      \item\label{alg:plugindec}\emph{Decentralized steps}\\
                        For $i=M+1$ execute Steps (\ref{alg:step1dec})-(\ref{alg:step3dec}) of Algorithm \ref{alg:practdec};
                        
                      \item\label{alg:plugindist}\emph{Distributed steps}\\
                        \begin{itemize}
                        \item For subsystem $\subss\Sigma{M+1}$, if $\delta_{M+1,j}=1$, compute the matrix $L_{M+1,j}$, $\forall j\in\NN_{M+1}$ solving $\min_{\substack{L_{M+1,j}}} \norme{\GG_{M+1}\bar A_{M+1,j}\GG_{j}}{p}$, $p=1$ or $p=F$, and then compute $\mu_{M+1,j}$;
                        \item For subsystems $\subss\Sigma{k}$, if $\delta_{k,M+1}=1$, compute the matrix $L_{k,M+1}$, $\forall k\in\CC_{M+1}$ solving $\min_{\substack{L_{k,M+1}}} \norme{\GG_k\bar A_{k,M+1}\GG_{M+1}}{p}$, $p=1$ or $p=F$, and then compute $\mu_{k,M+1}$;
                        \end{itemize}
                      \item\label{alg:plugincen}\emph{Centralized steps}\\
                        Execute steps (\ref{alg:step1cen})-(\ref{alg:step4cen}) of Algorithm \ref{alg:practdec}.
                      \end{enumerate}
                    \end{algorithm}

               \subsection{Unplugging operation}
                    Assume subsystem $\subss\Sigma{q}$, $q\in\MM$ is removed. We will show that no update of the DSE is required  in order to guarantee \eqref{eq:conv} and \eqref{eq:bounderror}. In the following, vectors, matrices and sets with a hat are quantities of the DSE after subsystem $q$  has been removed. As an example, the matrix 
                    $$
                    \hat T = \matr{\mu_{11} & \cdots & \mu_{1,q-1} & \mu_{1,q+1} & \cdots & \mu_{1,M} \\ 
                                              \vdots & \vdots & \vdots & \vdots & \vdots & \vdots \\ 
                                              \mu_{q-1,1} & \cdots & \mu_{q-1,q-1} & \mu_{q-1,q+1} & \cdots & \mu_{q-1,M} \\
                                              \mu_{q+1,1} & \cdots & \mu_{q+1,q-1} & \mu_{q+1,q+1} & \cdots & \mu_{q+1,M} \\ 
                                              \vdots & \vdots & \vdots & \vdots & \vdots & \vdots \\ 
                                              \mu_{M,1} & \cdots & \mu_{M,q-1} & \mu_{M,q+1} & \cdots & \mu_{M,M} \\ 
                                            }\in\Rset^{M-1\times M-1}
                    $$                    
                    is obtained from matrix $T$, by eliminating the $q$-th row and column. Next, we show Assumptions \ref{ass:thetaass}-\ref{enu:Tschur}, \ref{ass:thetaass}-\ref{enu:bartheta} and \ref{ass:thetaass}-\ref{enu:Theta0} are still verified after the removal of $\subss \Sigma q$. \\
Let $\GG=(V,\EE)$ be the coupling graph of \eqref{eq:model}, i.e. a directed graph where vertices in $V=1:M$ are associated to subsystems and $(i,j)\in\EE\Leftrightarrow i\in\NN_j$. In the sequel we assume $\GG$ is strongly connected (see Definition \ref{def:sc_digraph} in \ref{sec:appendixA}). Indeed, if this is not true, then \eqref{eq:model} can be represented as a directed acyclic graph $\mathbb{G}$ whose nodes represent strongly connected subgraphs. In this case, a DSE can be designed for each system corresponding to a subgraph starting from the roots of $\mathbb{G}$.
The next proposition concerns Assumption \ref{ass:thetaass}-\ref{enu:Tschur}.
                    \begin{prop}
                      \label{prop:unplug}
                      If the matrix $T\in\Rset^{M\times M}$ in \eqref{eq:thetadyn} is Schur, then also the matrix $\hat T$ is Schur.
                    \end{prop}
                    The proof of Proposition \ref{prop:unplug} can be found in \ref{sec:appendixA}.  The next result guarantees Assumption \ref{ass:thetaass}-\ref{enu:bartheta} still holds after the removal of subsystem $q$.
                    \begin{prop}
                      \label{prop:uniqueeq}
                      For $q\in\MM$, let $\hat\theta=(\theta_1,\ldots,\theta_{q-1},\theta_{q+1},\ldots,\theta_{M-1})$,\\  $\hat \alpha=(\alpha_1,\ldots,\alpha_{q-1},\alpha_{q+1},\ldots,\alpha_{M-1})$ and 
 \begin{equation}
                        \label{eq:hatTheta0}
                        \hat\Theta_0=\{\xi\in\Rset^{M-1} : (\xi_1,\ldots,\xi_{q-1},0,\xi_q,\ldots,\xi_{M-1})\in\Theta_0 \}
                      \end{equation}
                      If Assumption \ref{ass:thetaass}-\ref{enu:bartheta} holds, the unique equilibrium $\hat{\bar\theta}$ of system 
                      \begin{equation}
                        \label{eq:unpluggingthetadyn}
                        \hat\theta^+=\hat T\hat\theta+\hat\alpha
                      \end{equation}
                      is such that $\hat{\bar\theta}\in\hat\Theta_0$.
                    \end{prop}
                    The proof of Proposition \ref{prop:uniqueeq} can be found in \ref{sec:appendixB}.  Finally, the following proposition concerns Assumption \ref{ass:thetaass}-\ref{enu:Theta0}.
                    \begin{prop}
                      \label{prop:projinvariant}
                      For $q\in\MM $, the set 
                      \begin{equation}
                        \label{eq:hatTheta}
                        \hat\Theta=\{\hat\theta\in\Rset^{M-1} : (\hat{\theta}_1,\ldots,\hat{\theta}_{q-1},0,\hat{\theta}_q,\ldots,\hat{\theta}_{M-1})\in\Theta_\infty \}
                      \end{equation}
                      is an RPI set for system \eqref{eq:unpluggingthetadyn}.
                    \end{prop}
                    The proof of Proposition \ref{prop:projinvariant} can be found in \ref{sec:appendixC}.  From Proposition \ref{prop:projinvariant} we have that the projection of set $\Theta$ on the coordinates $\hat \theta$ is still an RPI set for \eqref{eq:model} after the removal of subsystem $q$. We also note that the set $\hat\Theta$ is not the maximal RPI, i.e. with a new execution of Step (\ref{alg:step3cen}) of Algorithm \ref{alg:practdec} we could obtain an RPI set $\hat\Theta_\infty$ verifying $\hat\Theta\subseteq\hat\Theta_\infty$. We also note that the projection  $\hat{\bar\Theta}$ of $\bar\Theta$ on the coordinates $\hat \theta$ is a box verifying $\hat{\bar\Theta}\subseteq\hat\Theta$. However, with a new execution of Step (\ref{alg:step4cen}) of Algorithm \ref{alg:practdec} we could obtain a bigger inner box approximation.

          \section{Examples}
               \label{sec:example}
               In this section, we apply the proposed distributed state estimator to a power network system composed by several power generation areas coupled through tie-lines. The dynamics of an area equipped with primary control and linearized around equilibrium value for all variables can be described by the following continuous-time LTI model \cite{Saadat2002}
               \begin{equation}
                 \label{eq:ltipower}
                 \subss{\Sigma}{i}^C:\quad\subss{\dot{x}}{i} = A_{ii}\subss x i + \bar B_{i}\subss{\bar u} i + \sum_{j\in\NN_i}A_{ij}\subss x j + \subss w i
               \end{equation}
               where $\subss x i=(\Delta\theta_i,~\Delta\omega_i,~\Delta P_{m_i},~\Delta P_{v_i})$ is the state, $\subss{\bar u} i = (\Delta P_{ref_i}, \Delta P_{L_i})$ is composed by the control input of each area and the local power load and $\NN_i$ is the sets of parent areas, i.e. areas directly connected to $\subss\Sigma i^C$ through tie-lines. In \eqref{eq:ltipower}, $\subss w i\in\Rset^{n_i}$ is the disturbance term for the $i$-th area and it is bounded in the polytopic set $\Wset_i\subset\Rset^{n_i}$. The matrices of system \eqref{eq:ltipower} are defined as
               \begin{equation*}
                 \label{eq:matrixpower}
                 \begin{aligned}
                   &A_{ii}(\{P_{ij}\}_{j\in\NN_i}) = \matr{ 0 & 1 & 0 & 0 \\ -\frac{\sum_{j\in\NN_i}{P_{ij}} }{2H_i} & -\frac{D_i}{2H_i} & \frac{1}{2H_i} & 0 \\ 0 & 0 & -\frac{1}{T_{t_i}}  & \frac{1}{T_{t_i}} \\ 0 & -\frac{1}{R_iT_{g_i}} & 0 & -\frac{1}{T_{g_i}} }\\
                   &\bar B_{i} = \matr{ 0 & 0 \\ 0 & -\frac{1}{2H_i} \\ 0 & 0 \\ \frac{1}{T_{g_i}} & 0 }\quad A_{ij} = \matr{ 0 & 0 & 0 & 0 \\ \frac{P_{ij}}{2H_i} & 0 & 0 & 0 \\ 0 & 0 & 0  & 0 \\ 0 & 0 & 0 & 0 }.
                 \end{aligned}
               \end{equation*}
               For the meaning of constants as well as parameter values we defer the reader to Section 1 of \cite{Riverso2012f}. We obtain models $\subss\Sigma i$ by discretizing models $\subss\Sigma i^C$ with $1~sec$ sampling time, using exact discretization and treating $\subss{\bar u} i$, $\subss x j,~j\in\NN_i$ and $\subss w i$ as exogenous signals. We note that using the proposed discretization scheme the set of neighbors $\NN_i$ does not change. In the following we propose different design of the distributed state estimator for a power network composed by four areas as in Figure \ref{fig:scenario1} (Scenario 1 of \cite{Riverso2012f}\footnote{For the simulations, we use the load power steps given in Section 1.1 of \cite{Riverso2012f} and the control inputs computed using MPC controllers as in Section 2 of \cite{Riverso2012f}.}).
               \begin{figure}[!htb]
                 \centering
                 \includegraphics[scale=0.55]{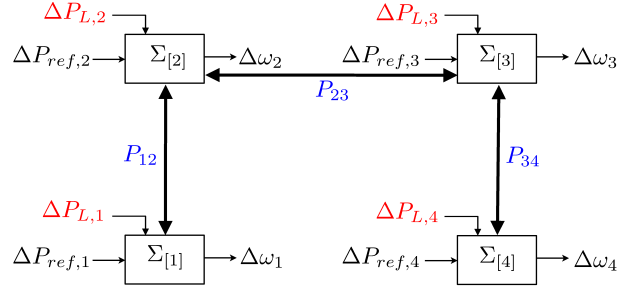}
                 \caption{Power network system composed by four areas}
                 \label{fig:scenario1}
               \end{figure}
               In Example $1$ and $2$, for each area, we consider the following bounds on the state estimation error
               \begin{equation}
                 \label{eq:errorConstraints}
                 \begin{aligned}
                   \Eset_i = \{ \subss e i\in\Rset^{n_i}:~&\norme{\subss e {i,1}}{\infty}\leq 0.005,~
                   \norme{\subss e {i,k}}{\infty}\leq 0.01,~k\in 2:4\}.
                 \end{aligned}
               \end{equation}
               We highlight that constraints \eqref{eq:errorConstraints} correspond in tolerating state estimation errors less then $10\%$ of the maximum value assumed by the state variables. In Example $3$, we consider constraints on the error equal to $ 2\Eset_i,~\forall i\in\MM$.
               
               \subsection{Example 1}
                    \label{sec:example1}
                    As first example, we consider $\delta_{ij}=1,~\forall i\in\MM,~\forall j\in\NN_i$, $\Wset_i=\{0\},~\forall i\in\MM$ (i.e. no disturbances act on the system) and assume to measure only the angular speed deviation $\subss{\Delta\omega}i$ of each area. Therefore, outputs of subsystem $i$ are given by
                    \begin{equation}
                      \label{eq:exe1}
                      \subss y i=C_i\subss x i,\qquad C_i=\matr{0 & 1 & 0 & 0}.
                    \end{equation}
                    In this case, Algorithm \ref{alg:practdec} stops in Step (\ref{alg:step1cen}) because the computed sets $\Sset_i$ are such that $T$ is not Schur. We highlight that from the results of Step (\ref{alg:step1dist}), one obtains the same results if parameters $\delta_{ij}$ are all set equal to zero. Indeed, for matrices $C_i$ in \eqref{eq:exe1}, it is impossible to reduce the magnitude of the coupling terms $\bar A_{ij}=A_{ij}+L_{ij}C_j$ by solving the optimization problems \eqref{eq:Lijopt}.
               
               \subsection{Example 2}
                    \label{sec:example2}
                    We consider $\Wset_i=\{0\},~\forall i\in\MM$, i.e. no disturbances act on the system, and we assume to measure both $\subss{\Delta\theta}i$ and $\subss{\Delta\omega}i$ of each area. Therefore the outputs are given by
                    \begin{equation}
                      \label{eq:exe2}
                      \subss y i=C_i\subss x i,\qquad C_i=\matr{1 & 0 & 0 & 0\\0 & 1 & 0 & 0}.
                    \end{equation}
                    First we consider $\delta_{ij}=0,~\forall i\in\MM,~\forall j\in\NN_i$. In this case, as in the first example, since we cannot take advantage of the knowledge of parents' outputs, Algorithm \ref{alg:practdec} stops before its conclusion. Indeed, it is impossible to find sets $\Sset_i$ such that $T$ is Schur. This example shows that if we also consider more output variables for each subsystem, Algorithm \ref{alg:practdec} can stop in Step (\ref{alg:step1cen}) due the magnitude of the coupling terms $A_{ij}$. Now we consider $\delta_{ij}=1,~\forall i\in\MM,~\forall j\in\NN_i$. In this case we can reduce the magnitude of the coupling terms. Solving optimization problems \eqref{eq:Lijopt}, we can compute matrices $L_{ij}$ such that $\bar A_{ij}=\Zero_{n_i\times n_j}$, hence the Schurness of matrix $T$ is guaranteed since sets $\Sset_i$ are $\lambda_i$-contractive. In this case, $T=\diag(0.932,~0.843,~0.711,~0.889)$ and $\bar\Theta=\{\theta\in\Rset^4:~0\leq\theta_i\leq 1,~\forall i=1:4\}$. We note that if matrix $T$ is diagonal, Step (\ref{alg:step4cen}) of Algorithm \ref{alg:practdec} can be skipped since $\Theta_\infty=\bar\Theta$. \\
                    We performed an estimation experiment initializing the local state estimators $\subss{\tilde\Sigma} i,~i\in\MM$ with $\subss \tx i (0)=\subss x i(0)-\subss e i(0)$, where $\subss e i(0)$ is a vertex of the set $\Sset_i$. In Figure \ref{fig:eNoW} we show the maximum state estimation error defined as
                    \begin{equation}
                      \label{eq:etilde}
                      \subss{\tilde e}j(t)=\max_{i\in\MM}\abs{\subss x {i,j}(t)-\subss \tx {i,j}(t)}
                    \end{equation}
                    where $\subss x {i,j}$ and $\subss \tx {i,j}$ are, respectively, the real and estimated state trajectory of the $j$-th state of the $i$-th subsystem.
                    \begin{figure}[!htb]
                      \centering
                      \includegraphics[scale=0.32]{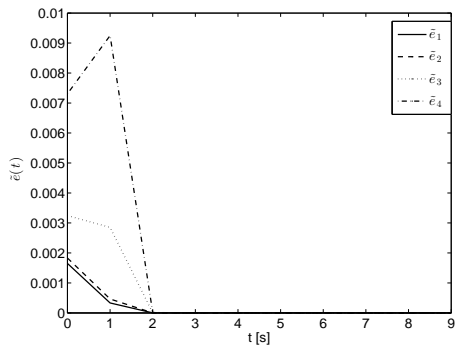}
                      \caption{Maximum estimation errors $\subss {\tilde e} j$ defined as in \eqref{eq:etilde}, for Example $2$.}
                      \label{fig:eNoW}
                    \end{figure}
                    From Figure \ref{fig:eNoW} we note that, since no disturbances act on the system, the state estimation error $\subss e i(t)$ converges to zero as $t\rightarrow\infty$, i.e. \eqref{eq:conv} is verified.

               \subsection{Example 3}
                    \label{sec:example3}
                    We consider $\Wset_i=\{ \subss w i\in\Rset^{n_i}:~\norme{\subss w {i,k}}{\infty}\leq 10^{-5},~k=1:4\},~\forall i\in\MM$ and output variables given in \eqref{eq:exe2}. As in Example $2$, by considering $\delta_{ij}=1,~\forall i\in\MM,~\forall j\in\NN_i$, Algorithm \ref{alg:practdec} does not stop at any intermediate step. We have performed a similar experiment as in Example $2$, but generating statistically independent random samples $\subss w i (t)$ from the uniform distribution on $\Wset_i$. In Figure \ref{fig:eYesW} and \ref{fig:eYesWzoom}, we show the maximum estimation error at the beginning of the experiment (Figure \ref{fig:eYesW}) and for $t\geq 10$ (Figure \ref{fig:eYesWzoom}). In particular, even in presence of disturbances on the system the state estimation error $\subss e i(t)$ lies in the set $\Eset_i$, $\forall i\in\MM$.
                    \begin{figure}[!htb]
                      \centering
                      \includegraphics[scale=0.32]{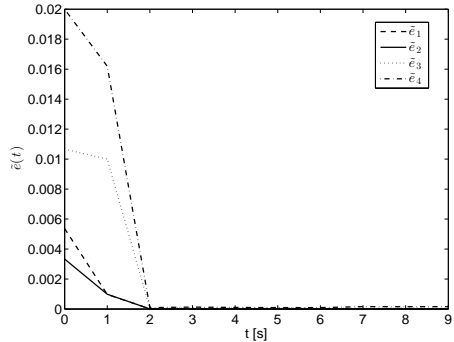}
                      \caption{Maximum estimation error $\subss{\tilde e} j(t)$, $t=0:9$ defined as in \eqref{eq:etilde}, for Example $3$.}
                      \label{fig:eYesW}
                    \end{figure}
                    \begin{figure}[!htb]
                      \centering
                      \includegraphics[scale=0.32]{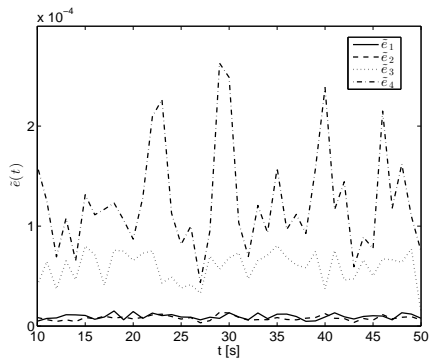}
                      \caption{Maximum estimation error $\subss{\tilde e} j(t)$, $t=10:50$ defined as in \eqref{eq:etilde}, for  Example $3$.}
                      \label{fig:eYesWzoom}
                    \end{figure}

               \subsection{Comparison with \cite{Farina2011b}}
                    In the previous examples we considered polytopic sets $\Eset_i,~i\in\MM$ defined in \eqref{eq:errorConstraints} that are also zonotopes, i.e. centrally symmetric polytopes. This is required by the DSE in \cite{Farina2011b}. Local state estimators in \cite{Farina2011b} depend on the state of parent systems, but not on their outputs. This corresponds to setting $\delta_{ij}=0$, $i,j\in\MM$ in our scheme. Using the DSE in \cite{Farina2011b}, we cannot compute the observers for the power network system in Examples $1$ and $2$. In fact, since all parameters $\delta_{ij}$ are zero it is impossible to reduce the magnitude of the coupling by using parents' outputs, as we do in Example $2$.

                    Moreover in \cite{Farina2011b}, the authors look for a family of sets $\Sset$ composed by mRPI sets $\Sset_i$ verifying
                    \begin{equation}
                      \label{eq:invariancefamFarinaScatto}
                      \bar A_{ii}\Sset_i\oplus\bigoplus_{j\in\NN_i}A_{ij}\Sset_j\subseteq\Sset_i\subseteq\Eset_i.
                    \end{equation}
                    We note that \eqref{eq:invariancefamFarinaScatto} is a special case of the parameterized RPI family in \eqref{eq:invariancefam} when $\delta_{ij}=0$ (i.e. $\bar A_{ij}=A_{ij}$) and $\Wset_i=\{0\}$. Since from \eqref{eq:invariancefamFarinaScatto} one has $\theta^+=\theta$, the matrix $T$ is not Schur, and for guaranteeing convergence of the estimates one has to check Schurness of the overall matrix $\mbf{A+LC}$.
In our distributed estimator, convergence of the estimates can be checked by testing the Schurness of $T$, i.e. no assumptions on the overall matrices are needed. 

          \section{Conclusion}
               \label{sec:conclusion}
               In this paper, we proposed a novel partition-based DSE for linear discrete-time subsystems affected by bounded disturbances. Our method guarantees convergence of the state estimates and the fulfillment of given bounds on state estimation errors. Similarly to \cite{Farina2011b} our DSE can be directly used together with state-feedback distributed control schemes such as \cite{Riverso2012e} although further research is needed for assessing the stability properties of the closed-loop system. In the future, we will also consider the problem of decentralizing completely computations required in the design process. This would lead to state-estimators that can be designed using local computational resources only, so coping with the plug-and-play design requirements of the model predictive control scheme proposed in \cite{Riverso2013c}.

    \appendix

          \section{Proof of Proposition \ref{prop:unplug}}
               \label{sec:appendixA}
                 The proof of Proposition \ref{prop:unplug} hinges on Perron-Frobenious theory for nonnegative matrices. Next, we provide relevant definition,  deferring the reader to \cite{Meyer2000} for further details.
                 \begin{defi}
                   The graph, $\Gamma(Q)=(V,\EE)$ of $Q\in\Rset^{M\times M}$ is the directed graph with nodes $V=1:M$ and edges $\EE=\{(i,j): q_{ij}\neq 0\}$ where $q_{ij}$ is the $ij$-th element of the matrix $Q$.
                 \end{defi}
                 \begin{defi}
\label{def:sc_digraph}
                   A directed graph $\Gamma$ is strongly connected if for any pair of nodes $(N_i,N_j)$ there exists a sequence of edges which leads from $N_i$ to $N_j$.
                 \end{defi}
                 \begin{defi}
                   A matrix $Q\in\Rset^{M\times M}$ is irreducible if there is no permutation matrix $P$ such that
                   $$
                   Z = PQP^T= \matr{ Q_{11} & Q_{12} \\ 0 & Q_{22} },
                   $$
                   where $Q_{11}\in\Rset^{q\times q}$, $Q_{22}\in\Rset^{M-q,M-q}$ and $Q_{12}\in\Rset^{q,M-r}$, $0<q<M$. 
                 \end{defi}
\begin{proof}[Proof of Proposition~ \ref{prop:unplug}]
The matrix $T$ in \eqref{eq:thetadyn} is nonnegative, i.e. $\mu_{ij}\geq 0,~\forall~i,j\in\MM$. Moreover, $\GG=\Gamma(T)$ and since $\GG$ is strongly connected, $T$ is irreducible \cite[p. 671]{Meyer2000}. Let 
\begin{equation}
                   \label{eq:taumatrix}
                   \T = \matr{\mu_{11} & \cdots & \mu_{1,q-1} & 0 & \mu_{1,q+1} & \cdots & \mu_{1,M} \\ 
                                      \vdots & \vdots & \vdots & \vdots & \vdots & \vdots & \vdots \\ 
                                      \mu_{q-1,1} & \cdots & \mu_{q-1,q-1} & 0 & \mu_{q-1,q+1} & \cdots & \mu_{q-1,M} \\
                                      0 & \cdots & 0 & 0 & 0 & \cdots & 0 \\
                                      \mu_{q+1,1} & \cdots & \mu_{q+1,q-1} & 0 & \mu_{q+1,q+1} & \cdots & \mu_{q+1,M} \\ 
                                      \vdots & \vdots & \vdots & \vdots & \vdots & \vdots & \vdots \\ 
                                      \mu_{M,1} & \cdots & \mu_{M,q-1} & 0 & \mu_{M,q+1} & \cdots & \mu_{M,M} \\ 
                                    }\in\Rset^{M\times M}
                 \end{equation}

 From Weilandt's Theorem \cite[p. 675]{Meyer2000}, one has $\rho(\T)\leq\rho(T)$. Moreover, up to a permutation matrix, one has $\T=\matr{ 0 & 0 \\ 0 & \hat T }$ and hence $\rho(\hat T)\leq \rho(\T)$. Therefore $\rho(\hat T)\leq\rho(T)$ and the proof is concluded recalling that, by assumption, $\rho(T)<1$.
\end{proof}

          \section{Proof of Proposition \ref{prop:uniqueeq}}
               \label{sec:appendixB}
                First we define matrix $\T$ as in \eqref{eq:taumatrix}.
                 Since $\mu_{ij}\geq 0,~\forall i,j\in\MM$, the elements of matrices $\T^k$ and $T^k$ are nonnegative $\forall k\geq 0$. Moreover we can show that the $ij$-th element of $\T^k$ (with abuse of notation, $\T_{ij}^k$) is smaller than the $ij$-th element of $T^k$ (with abuse of notation, $T_{ij}^k$), i.e. $\T_{ij}^k\leq T_{ij}^k$, $\forall i,j\in\MM$ and $\forall k\geq 0$. 
                 Let $\bar\tau=(\hat{\bar\theta}_1,\ldots,\hat{\bar\theta}_{q-1},0,\hat{\bar\theta}_q,\ldots,\hat{\bar\theta}_{M-1})\in\Rset^M$, where $\hat{\bar\theta}_i$ is the $i$-th component of vector $\hat{\bar\theta}$ and
                 \begin{equation}
                   \label{eq:tildealpha}
                   \tilde\alpha=\alpha\in\Rset_+^M,~ \tilde\alpha_q=0.
                 \end{equation}
                 The unique equilibrium point of system \eqref{eq:thetadyn} can be written as $\bar\theta=\sum_{k=0}^\infty T^k\alpha$. Moreover from \eqref{eq:unpluggingthetadyn} and the definitions of $\bar\tau$ and $\T$, we have that $\bar\tau=\sum_{k=0}^\infty \T^k\tilde\alpha$. Since $\T_{ij}^k\leq T_{ij}^k$, then $\sum_{k=0}^\infty\T_{ij}^k\leq\sum_{k=0}^\infty T_{ij}^k$ and hence $\bar\tau\leq\bar\theta$ element-wise. From \eqref{eq:Theta0comp} and Assumption~\ref{ass:thetaass}-\ref{enu:bartheta}, one has
 $\bar\tau\in\Theta_0$. Therefore, we can conclude that $\hat{\bar\theta}\in\hat\Theta_0$.
          
          \section{Proof of Proposition \ref{prop:projinvariant}}
               \label{sec:appendixC}                 
               \begin{proof}[Proof of Proposition~ \ref{prop:projinvariant}]
                 After subsystem $q$ has been removed, the dynamics of contraction factors $\hat \theta$ is given by \eqref{eq:unpluggingthetadyn}.
                 In the following we show that $\hat\Theta$ defined in \eqref{eq:hatTheta} is an RPI set for  \eqref{eq:unpluggingthetadyn}.
                 
                 From the invariance of set $\Theta_\infty$ we have that $T\theta+\alpha\in\Theta_\infty,~\forall \theta\in \Theta_\infty$. Moreover, since $0\in\Theta_\infty$, we have  $\T\theta+\tilde\alpha\in\Theta_\infty,~\forall\theta\in\Theta_\infty$ (where $\T$ and $\tilde\alpha$ are defined in \eqref{eq:taumatrix} and \eqref{eq:tildealpha}) i.e. the set $\Theta_\infty$ is also invariant for the LTI system  $\theta^+=\T\theta+\tilde\alpha$ and the $q$-th component of $\theta$ is always zero. Therefore, we can conclude that the projection of set $\Theta_\infty$ defined in \eqref{eq:hatTheta} is an RPI set for system $\hat\theta^+=\hat T\hat\theta+\hat\alpha$.
               \end{proof}

                 \bibliographystyle{IEEEtran}
                 \bibliography{distributed_state_estimator-report}

\end{document}

%% file: macro.tex
\def\Rset{\mathbb{R}}
\def\Nset{\mathbb{N}}

\newcommand{\CC}{{\mathcal C}}

\newcommand{\EE}{{\mathcal E}}
\newcommand{\FF}{{\mathcal F}}
\newcommand{\GG}{{\mathcal G}}
\newcommand{\HH}{{\mathcal H}}

\newcommand{\MM}{{\mathcal M}}
\newcommand{\NN}{{\mathcal N}}

\newcommand{\PP}{{\mathcal P}}

\newcommand{\SSS}{{\mathcal S}}
\newcommand{\T}{{\mathcal T}}

\newcommand{\mbf}[1]{\mathbf{#1}}
\newcommand{\px}{{x^+}}
\newcommand{\subss}[2]{#1_{[#2]}}
\newcommand{\tx}{{\tilde x}}

\newcommand{\diag}{\mbox{diag}}

\newcommand{\Xset}{\mathbb{X}}

\newcommand{\Wset}{\mathbb{{{W}}}}

\newcommand{\abs}[1]{{|{#1}|}}
\newcommand{\norme}[2]{{||{#1}||_{#2}}}

\newcommand{\One}{\textbf{1}}
\newcommand{\Zero}{\textbf{0}}

\newcommand{\tpx}{{\tilde{x}}^+}
\newcommand{\pe}{{e}^+}
\newcommand{\Eset}{\mathbb{E}}
\newcommand{\Sset}{\mathbb{S}}
\newcommand{\dist}[2]{\mbox{dist}({#1},{#2})}

\newcommand{\imply}{\Rightarrow}

\newcommand{\ba}[1]{\begin{array}{#1}}
\newcommand{\ea}{\end{array}}

\newcommand{\matr}[1]{
\begin{bmatrix}
    #1
\end{bmatrix}
}

